\documentclass[11pt,a4paper,final]{article}
\usepackage{geometry}
\usepackage{amsmath,amsfonts,amssymb,amsthm,version}
\usepackage{mathrsfs,fancybox,pifont}
\usepackage{graphicx}
\usepackage{url,hyperref}
\usepackage[notcite,notref]{showkeys}
\usepackage{color}
\usepackage{subfigure,multirow}
\usepackage{epstopdf}
\usepackage{cases}
\usepackage{mathtools}
\usepackage{algorithm,algorithmic}
\usepackage{authblk}
\usepackage{float} 
\usepackage{subfigure}

\allowdisplaybreaks

\numberwithin{equation}{section}
\numberwithin{figure}{section}
\numberwithin{table}{section}

\theoremstyle{plain}
\newtheorem{theorem}{Theorem}[section]

\newtheorem{proposition}[theorem]{Proposition}

\theoremstyle{definition}
\newtheorem{definition}[theorem]{Definition}

\newtheorem{lemma}[theorem]{Lemma}

\theoremstyle{remark}

\title{Bell Polynomial Approach and Wronskian Technique \\to  Good Boussinesq Equation }
\author[1]{Xiaotian Dai\footnote{21210180083@m.fudan.edu.cn}}
\author[1]{Zhenyun Qin\footnote{zyqin@fudan.edu.cn}}
\affil[1]{School of Mathematical Sciences, Fudan University, Shanghai200433, PR China}

\date{}

\begin{document}

\maketitle

{\linespread{1.0} \selectfont
\begin{abstract}
The elementary and systematic binary Bell polynomial approach is applied to the good Boussinesq equation. The bilinear representation, $n$-soliton solutions, bilinear B\"acklund transformation, Lax pair and infinite conservation laws of the good Boussinesq equation are obtained directly. In addition, from the reduction conditions of the obtained Lax pairs, the $n$ order $\operatorname{Wronskian}$ determinant solution of the equation is also constructed.

  \medskip 
  \noindent{\bf Keywords}: good Boussinesq equation, binary Bell polynomials, Lax pairs, infinite conservation laws, $\operatorname{Wronskian}$ determinant solution

\end{abstract}
}
\linespread{1.0} \selectfont

\section{Introduction}

            As the soliton phenomena were first observed in 1834 \cite{Drazin} and Korteweg-de Vries (KdV) equation was solved by the inverse scattering method \cite{AAC}, finding soliton solutions of nonlinear PDEs has become one of the most exciting and extremely active areas of research. Investigation of integrability for a soliton equation can be regarded as a pretest and the first step of its exact solvability. Among the direct algebratic methods employed to study the integrability of soliton equations, the Hirota method has been a particularly powerful method \cite{Hirota}. Once a given soliton equation is written in bilinear form, on one hand, such results as multi-soliton solutions, quasi periodic wave solutions and other exact solutions are usually obtained, and on the other hand, the integrable properties of the soliton equation, such as the B\"acklund transformation (BT) and Lax pair can also be investigated. However, the construction of bilinear form and bilinear BT of the original soliton equation is not as one would wish. It relies on a particular skill in using appropriate dependent variable transformation, exchange formulas and bilinear identities.

        Recently, Lambert and his co-workers have proposed an alternative procedure based on the use of Bell polynomials to obtain bilinear forms, bilinear BTs, Lax pairs for soliton equations in a lucid and systematic way \cite{GLN, LLS, LS}.
         Fan developed this method to find infinite conservation laws of soliton equations and proposed the super Bell polynomials \cite{egfan, FH}. Ma systematically analyzed the connection between Bell polynomials and new bilinear equation \cite{MaWX}.
         Chen et al. proposed the Maple automatic program to construct B\"acklund transformation, Lax pairs and the infinite conservation laws for nonlinear evolution equations \cite{MWCY}. 

         In Hirota's method, the final step is to prove the general $n$-soliton \cite{Hirota1971}. However, this proof can be quite difficult and since the early 1980s, it has been more usual to re-express the solutions found in terms of a $\operatorname{Wronskian}$  or Grammian determinant, or Pfaffians of several types, and verify the solution in that form \cite{Hirota}. The $\operatorname{Wronskian}$ technique provides direct verifications of solutions to bilinear equations by taking the advantage that special structure of a $\operatorname{Wronskian}$ contributes simple forms of its derivatives \cite{FN}. However, the key element of this technique is to figure out a system of linear differential conditions, which ensures that the $\operatorname{Wronskian}$ formulation presents solutions of the Hirota bilinear form. 

         In using the $\operatorname{Wronskian}$ technique to solve the KdV equation \cite{MY, Matveev}
         \begin{equation*}
              {{u}_{t}}-6u{{u}_{x}}-{{u}_{xxx}}=0,
         \end{equation*}
         one usually starts from its Lax pair
         \begin{equation*}
             {{\varphi }_{xx}}=\left( \lambda -u \right)\varphi , \quad {{\varphi }_{t}}={{\varphi }_{xxx}}+3\left( \lambda +u \right){{\varphi }_{x}}. 
         \end{equation*}
         Choose $u=0, \lambda =\frac{{{k}^{2}}}{4}$, then its Lax pair can be reduced to 
        \begin{equation*}
            {{\varphi }_{xx}}=\frac{1}{4}{{k}^{2}}\varphi , \quad {{\varphi }_{t}}=4{{\varphi }_{xxx}}.
         \end{equation*}
        Therefore, choose  ${{\varphi }_{j}}\in {{C}^{\infty }}\left( \Omega  \right)$, satisfying
         \begin{equation*}
            {{\varphi }_{j,xx}}=\frac{k_{j}^{2}}{4}{{\varphi }_{j}}, \quad {{\varphi }_{j,t}}=4{{\varphi }_{j,xxx}}, \quad j=1,2,\cdots ,n, 
         \end{equation*}
         which constitutes the $\operatorname{Wronskian}$ condition of KdV equation.      

       Inspired by Ma's work \cite{MY}, the purpose of this paper is to extend the binary Bell polynomial approach and $\operatorname{Wronskian}$ technique to the good Boussinesq equation:
        \begin{equation}\label{eq101}
        {{u}_{tt}}-{{u}_{xx}}+{{\left( {{u}^{2}} \right)}_{xx}}+\frac{1}{3}{{u}_{xxxx}}=0,
        \end{equation}
        which is obviously equivalent to the following well-posed Boussinesq equation under the scaling transformation $u\to U, x\to \frac{1}{\sqrt{3}}X, t\to \frac{1}{\sqrt{3}}T$.
        \begin{equation}\label{eq102}
        {{U}_{TT}}-{{U}_{XX}}+{{U}_{XXXX}}+{{\left( {{U}^{2}} \right)}_{XX}}=0.
        \end{equation}

        Equation (\ref{eq102}) is of fundamental physical interest, and it occurs in a wide variety of physical systems \cite{Whitham, Xu, Karpman, Turitsyn}, because it describes the lowest order (in terms of wave amplitudes) nonlinear effects in the evolution of perturbations with the dispersion relation close to that for sound waves. It also appears in the problems dealing with propagations of nonlinear waves in a medium with positive dispersion, for example, electromagnetic waves interacting with transversal optical phonons in nonlinear dielectrics \cite{Xu}, magnetosound waves in plasmas \cite{Karpman} and magnetoelastic waves in antiferromagnets \cite{Turitsyn}. Furthermore, it is known that for the transonic speed perturbations, by neglecting the interaction of waves moving in the opposite directions.

        So far, the soliton solutions and multicollapse solutions were investigated by using Darboux transformations \cite{Matveev} and Hirota's method \cite{HirotaR}. Li et al. \cite{Chunxia} have obtained solitons, negations, positions and complexitons for Eq.(\ref{eq101}). However, not much work has been done on the integrability  of Eq.(\ref{eq101}).
        
        The problem we concern in this paper is to study the integrability of Eq.(\ref{eq101}), and additionally figure out the usual method for constructing the $\operatorname{Wronskian}$ condition.
        
        This paper is organized as follows. In Section2, we briefly present necessary notations on binary Bell polynomials that will be used in this paper. These results will then be applied to construct the bilinear representations, $n$-soliton solutions, bilinear BT, Lax pair and infinite conservation laws to Eq.(\ref{eq101}) in Section 3. In Section 4, we construct $\operatorname{Wronskian}$ condition of Eq.(\ref{eq101}) by using the Lax pairs obtained in Section 3.3, and the solutions in $\operatorname{Wronskian}$ form are verified. Finally, Section 5 presents our conclusions.

\section{Binary Bell Polynomials}
        To make our presentation easy understanding and self-contained, we first fix some necessary notations on the Bell polynomials, the details refer to Bell, Lambert and Gilson's work \cite{GLN, LLS, LS}.
        
        \begin{definition}
        Suppose that $y=y\left( {{x}_{1}},\cdots ,{{x}_{n}} \right)$ is a ${{C}^{\infty }}$ function with $n$ independent variables and we denote 
            \begin{equation*}
                {{y}_{{{r}_{1}}{{x}_{1}},\cdots ,{{r}_{l}}{{x}_{l}}}}=\partial _{{{x}_{1}}}^{{{r}_{1}}}\cdots \partial _{{{x}_{l}}}^{{{r}_{l}}}y, \quad {{r}_{1}}=0,\cdots ,{{n}_{1}},\cdots , {{r}_{l}}=0,\cdots ,{{n}_{l}},
            \end{equation*}
        where $l$ denotes arbitrary integers, then 
        \begin{equation}\label{eq201}
                 {{Y}_{{{n}_{1}}{{x}_{1}},\cdots ,{{n}_{l}}{{x}_{l}}}}\left( y \right)\equiv {{Y}_{{{n}_{1}},\cdots ,{{n}_{l}}}}\left( {{y}_{{{r}_{1}}{{x}_{1}},\cdots ,{{r}_{l}}{{x}_{l}}}} \right)={{e}^{-y}}\partial _{{{x}_{1}}}^{{{n}_{1}}}\cdots \partial _{{{x}_{l}}}^{{{n}_{l}}}{{e}^{y}}  
                \end{equation}
        is a polynomial in the partial derivatives of $y$ with respect to ${{x}_{1}},\cdots ,{{x}_{l}}$, which called a multi-dimensional Bell polynomial. For the special case $f=f\left( x,t \right)$, the associated two-dimensional Bell polynomials defined by (\ref{eq201}) read
        \begin{equation}
            {{Y}_{x,t}}\left( f \right)={{f}_{x,t}}+{{f}_{x}}{{f}_{t}}, \quad {{Y}_{2x,t}}\left( f \right)={{f}_{2x,t}}+{{f}_{2x}}{{f}_{t}}+2{{f}_{x,t}}{{f}_{x}}+f_{x}^{2}{{f}_{t}}, \cdots
        \end{equation}
        \end{definition}

        By virtue of the above multi-dimensional Bell polynomials, the multi-dimensional binary Bell polynomials can be defined as follows.
        \begin{definition}
        Suppose that $w=f+g$, $v=f-g$, then 
            \begin{equation}
               {{\mathcal{Y}}_{{{n}_{1}}{{x}_{1}},\cdots ,{{n}_{l}}{{x}_{l}}}}(v,w)={{Y}_{{{n}_{1}},\cdots ,{{n}_{l}}}}\left( y \right)\left| _{{{y}_{{{r}_{1}}{{x}_{1}},\cdots ,{{r}_{l}}{{x}_{l}}}}=\left\{ \begin{matrix}
                {{v}_{{{r}_{1}}{{x}_{1,\cdots ,}}{{r}_{l}}{{x}_{l}}}}, {{r}_{1}}+\cdots +{{r}_{l}}\mbox { is odd}  \\
               {{w}_{{{r}_{1}}{{x}_{1}},\cdots {{r}_{l}}{{x}_{l}}}},  {{r}_{1}}+\cdots +{{r}_{l}}\mbox { is even}  \\
              \end{matrix}\right.} \right.           
          \end{equation}
          where the vertical line means that the elements on the left-hand are chosen according to the rule on the right-hand side.
        \end{definition}

        For example, the first few lowest-order binary Bell polynomials are 
        \begin{equation}
        \begin{aligned}
            &{{\mathcal{Y}}_{x}}\left( v,w \right)={{v}_{x}}, \quad {{\mathcal{Y}}_{2x}}\left( v,w \right)={{w}_{2x}}+v_{x}^{2}, \quad {{\mathcal{Y}}_{x,t}}\left( v,w \right)={{w}_{x,t}}+{{v}_{x}}{{v}_{t}}, \\
            & {{\mathcal{Y}}_{3x}}\left( v,w \right)={{v}_{3x}}+3{{w}_{2x}}{{v}_{x}}+v_{x}^{3},  \quad {{\mathcal{Y}}_{2x,t}}\left( v,w \right)={{v}_{2x,t}}+2{{w}_{x,t}}{{v}_{x}}+v_{x}^{2}{{v}_{t}}+{{w}_{2x}}{{v}_{t}}, \cdots
        \end{aligned}
        \end{equation}

         \begin{proposition}\label{pr203}
            The relations between the binary Bell polynomials and the standard Hirota D-operators can be given by the identity 
            \begin{equation}\label{eq205}
                {{\mathcal{Y}}_{{{n}_{1}}{{x}_{1}},\cdots ,{{n}_{l}}{{x}_{l}}}}\left( v=\ln {F}/{G,w=\ln FG}\; \right)={{\left( FG \right)}^{-1}}D_{{{x}_{1}}}^{{{n}_{1}}}\cdots D_{{{x}_{l}}}^{{{n}_{l}}}F\cdot G,
            \end{equation}
            where ${{n}_{1}}+{{n}_{2}}+\cdots +{{n}_{l}}\ge 1,$  and the Hirota D-operators are defined by
            \begin{equation}
                D_{{{x}_{1}}}^{{{n}_{1}}}\cdots D_{{{x}_{l}}}^{{{n}_{l}}}F\cdot G={{\left( {{\partial }_{{{x}_{1}}}}-{{\partial }_{x_{1}^{'}}} \right)}^{{{n}_{1}}}}\cdots {{\left( {{\partial }_{{{x}_{l}}}}-{{\partial }_{x_{l}^{'}}} \right)}^{{{n}_{1}}}}F\left( {{x}_{1}},\cdots ,{{x}_{l}} \right)G\left( x_{1}^{'},\cdots ,x_{l}^{'} \right)\left| _{x_{1}^{'}={{x}_{1}},\cdots ,x_{l}^{'}={{x}_{l}}}. \right.
            \end{equation}
         \end{proposition}
         In particular, when $F=G$, formula (\ref{eq205}) can be rewritten as 
         \begin{equation}\label{eq207}
             {{G}^{-2}}D_{{{x}_{1}}}^{{{n}_{1}}}\cdots D_{{{x}_{l}}}^{{{n}_{l}}}G\cdot G  ={\mathcal{Y}_{{{n}_{1}}{{x}_{1}},\cdots ,{{n}_{l}}{{x}_{l}}}}\left( 0,q=2\ln G \right) 
              = \left\{ \begin{matrix}
   0&, &{{n}_{1}}+\cdots +{{n}_{l}}\mbox{ is odd},  \\
   {{P}_{{{n}_{1}}{{x}_{1}},\cdots ,{{n}_{l}}{{x}_{l}}}}\left( q \right)&,& {{n}_{1}}+\cdots +{{n}_{l}}\mbox{ is even}.  \\
\end{matrix} \right.
         \end{equation}
        which is also called a P-polynomial
        \begin{equation}
            {{P}_{{{n}_{1}}{{x}_{1}},\cdots ,{{n}_{l}}{{x}_{l}}}}\left( q \right)={{\mathcal{Y}}_{{{n}_{1}}{{x}_{1}},\cdots ,{{n}_{l}}{{x}_{l}}}}\left( 0,q=2\ln F \right).
        \end{equation}
        which vanishes unless $\sum\limits_{i=1}^{l}{{{n}_{i}}}$ is even.

        The first few P-polynomials and  are
        \begin{equation}\label{eq209}
            {{P}_{x,t}}\left( q \right)={{q}_{x,t}}, \quad
            {{P}_{2x}}\left( q \right)={{q}_{2x}}, \quad
            {{P}_{3x,t}}\left( q \right)={{q}_{3x,t}}+3{{q}_{x,t}}{{q}_{2x}}, \quad
            {{P}_{4x}}\left( q \right)={{q}_{4x}}+3q_{2x}^{2}, \cdots
        \end{equation}
        Formula (\ref{eq205}) and (\ref{eq207}) will prove particularly useful in connecting nonlinear equations with their corresponding bilinear equations. This means that once a nonlinear equation can be expressed as a linear combination of P-polynomials, then its bilinear equation can be established directly.

         \begin{proposition}\label{pr204}
            The binary Bell polynomials ${{\mathcal{Y}}_{{{n}_{1}}{{x}_{1}},\cdots ,{{n}_{l}}{{x}_{l}}}}(v,w)$ can be written as a combination of P-polynomials and Y-polynomials
             \begin{equation}\label{eq210}
     \begin{aligned}        
	      {{\left( FG \right)}^{-1}}D_{{{x}_{1}}}^{{{n}_{1}}}\cdots D_{{{x}_{l}}}^{{{n}_{l}}}F\cdot G &= {{\mathcal{Y}}_{{{n}_{1}}{{x}_{1}},\cdots ,{{n}_{l}}{{x}_{l}}}}\left( v,w \right)\left| _{v=\ln \frac{F}{G},w=\ln FG} \right.  \\
		 & ={{\mathcal{Y}}_{{{n}_{1}}{{x}_{1}},\cdots ,{{n}_{l}}{{x}_{l}}}}\left( v,v+q \right)\left| _{v=\ln \frac{F}{G},q=2\ln G} \right. \\ 
 & =\sum\limits_{{{r}_{1}}=0}^{{{n}_{1}}}{\cdots }\sum\limits_{{{r}_{l}}=0}^{{{n}_{l}}}{\prod\limits_{i=1}^{l}{\left( \begin{matrix}
   {{n}_{i}}  \\
   {{r}_{i}}  \\
\end{matrix} \right)}}{{P}_{{{r}_{1}}{{x}_{1}},\cdots ,{{r}_{l}}{{x}_{l}}}}\left( q \right)\cdot {{Y}_{\left( {{n}_{1}}-{{r}_{1}} \right){{x}_{1}},\cdots ,\left( {{n}_{l}}-{{r}_{l}} \right){{x}_{l}}}}\left( v \right), 
    \end{aligned}
       \end{equation}
       where $\sum\limits_{i=1}^{l}{{{n}_{i}}}$ is even.
         \end{proposition}
         
         \begin{proposition}\label{pr205}
             In order to obtain the Lax pairs of corresponding NLEEs, we introduce the Hopf-Cole transformation $v=\ln \psi $, then the Y-polynomials can be written as
            \begin{equation}\label{eq211}
           {{Y}_{{{n}_{1}}{{x}_{1}},\cdots ,{{n}_{l}}{{x}_{l}}}}\left( v \right)\left| _{v=\ln \psi } \right.=\frac{{{\psi }_{{{n}_{1}}{{x}_{1}},\cdots ,{{n}_{l}}{{x}_{l}}}}}{\psi },
             \end{equation}
        which provides the shortest way to the associated Lax systems of NLEEs.
         \end{proposition}

    \section{Integrability of good Boussinesq equation}
    \subsection{Biinear representation}
    Using the homogeneous balance principle, in order to make the nonlinear term $u_x^2$ (or $u {{u}_{xx}}$) and the highest derivative term ${{u} _ {xxxx}} $ in the equation (\ref{eq101})  be partially balanced, we suppose that the highest degree of derivetives of $u\left( x,t \right)$ with respect to $x, t$ is $m$, then the highest degree of derivetives of ${{u}_{xxxx}}$ and $u_x^2$ with respect to $x, t$ is $m+4$ and $2m+2$ respectively. Let
    \begin{equation}
        m+4=2m+2,
    \end{equation}
    we get $m=2$, which shows that a dimensionless field $q$ can be related to the field $u$ by setting
    \begin{equation}\label{eq302}
        u=c{{q}_{xx}},
    \end{equation}
    with $c$ being a free constant to be the appropriate choice such that equation (\ref{eq101}) links with P-polynomials. Substituting (\ref{eq302}) into (\ref{eq101}), we can write the resulting equation in the form
    \begin{equation}\label{eq303}
        c{{q}_{xxtt}}-c{{q}_{xxxx}}+{{\left( {{c}^{2}}q_{xx}^{2} \right)}_{xx}}+\frac{1}{3}c{{q}_{6x}}=0.
    \end{equation}\label{eq304}
    Further integrating Eq.(\ref{eq303}) with respect to x twice yields
    \begin{equation}
        E\left( q \right)\equiv{{q}_{tt}}-{{q}_{xx}}+cq_{xx}^{2}+\frac{1}{3}{{q}_{4x}}=0.
    \end{equation}
    Comparing the third term of this equation with the formula (\ref{eq209}) implies that we should require $c=1$. The result equation is then cast into a combination form of P-polynomials
    \begin{equation}\label{eq305}
        E\left( q \right)\equiv{{P}_{2t}}(q)-{{P}_{2x}}(q)+\frac{1}{3}{{P}_{4x}}(q)=0.
    \end{equation}
    Making a change of dependent variable 
    \begin{equation*}
        u={{q}_{2x}}=2{{\left( \ln F \right)}_{2x}}
    \end{equation*}
    and noting the proposition \ref{pr203}, Eq.(\ref{eq305}) gives the bilinear representation as follows 
    \begin{equation}\label{eq306}
        \left( D_{t}^{2}-D_{x}^{2}+\frac{1}{3}D_{x}^{4} \right)F\cdot F=0.
    \end{equation}

    \subsection{$n$-soliton solutions}
    Once the bilinear representation of Eq.(\ref{eq101}) is given, associated soliton solutions are easily solved with the help of the perturbation expansion method. Here we leave out the computational process and give the $n$-soliton solutions directly, because its verification process is complicated \cite{Nguyen}.
    \begin{equation}
         u=2{{\left[ \log \sum\limits_{{{\mu }_{i}}\in \left\{ 0,1 \right\}}{\exp \left( \sum{{{\mu }_{i}}{{\eta }_{i}}+\sum\limits_{1\le i<j}{{{\mu }_{i}}{{\mu }_{j}}{{A}_{ij}}}} \right)} \right]}_{xx}}, \quad  
     \end{equation}
     where
     ${{\eta }_{i}}={{k}_{i}}x+{{w}_{i}}t+\eta _{i}^{\left( 0 \right)}, \quad w_{i}^{2}=k_{i}^{2}\left( 1-\frac{1}{3}k_{i}^{2} \right),$ and $$ {{A}_{ij}}=\log \left| -\frac{{{\left( {{w}_{i}}-{{w}_{j}} \right)}^{2}}-{{\left( {{k}_{i}}-{{k}_{j}} \right)}^{2}}+\frac{1}{3}{{\left( {{k}_{i}}-{{k}_{j}} \right)}^{4}}}{{{\left( {{w}_{i}}+{{w}_{j}} \right)}^{2}}-{{\left( {{k}_{i}}+{{k}_{j}} \right)}^{2}}+\frac{1}{3}{{\left( {{k}_{i}}+{{k}_{j}} \right)}^{4}}} \right|, i<j, i,j=1,2,3,\cdots$$
     
     For $n=1$, the single-soliton solution of the good Boussinesq equation(\ref{eq101}) can be written as 
     \begin{equation}      
	      u = 2{{\left[ \ln \left( 1+{{e}^{{{\eta }_{1}}}} \right) \right]}_{xx}}  = \frac{k_{1}^{2}}{2}{{\operatorname{sech}}^{2}}\frac{{{\eta }_{1}}}{2}, 
    \end{equation}
    where ${{\eta }_{1}}={{k}_{1}}x+{{w}_{1}}t+\eta _{1}^{0}$, \quad $w_{1}^{2}=k_{1}^{2}\left( 1-\frac{1}{3}k_{1}^{2} \right)$.
    
    For $n=2$, the two-soliton solution reads
    \begin{equation}
        u=2{{\left[ \log \left( 1+{{e}^{{{\eta }_{1}}}}+{{e}^{{{\eta }_{2}}}}+{{e}^{{{\eta }_{1}}+{{\eta }_{2}}+{{A}_{12}}}} \right) \right]}_{xx}},
    \end{equation}
    where ${{\eta }_{i}}={{k}_{i}}x+{{w}_{i}}t+\eta _{i}^{\left( 0 \right)}$, $w_{i}^{2}=k_{i}^{2}\left( 1-\frac{1}{3}k_{i}^{2} \right)$, $i=1,2$,  ${{e}^{{{A}_{12}}}}=-\frac{{{\left( {{w}_{1}}-{{w}_{2}} \right)}^{2}}-{{\left( {{k}_{1}}-{{k}_{2}} \right)}^{2}}+\frac{1}{3}{{\left( {{k}_{1}}-{{k}_{2}} \right)}^{4}}}{{{\left( {{w}_{1}}+{{w}_{2}} \right)}^{2}}-{{\left( {{k}_{1}}+{{k}_{2}} \right)}^{2}}+\frac{1}{3}{{\left( {{k}_{1}}+{{k}_{2}} \right)}^{4}}}$.

\begin{figure}[htbp]
\centering
\includegraphics[scale=0.8]{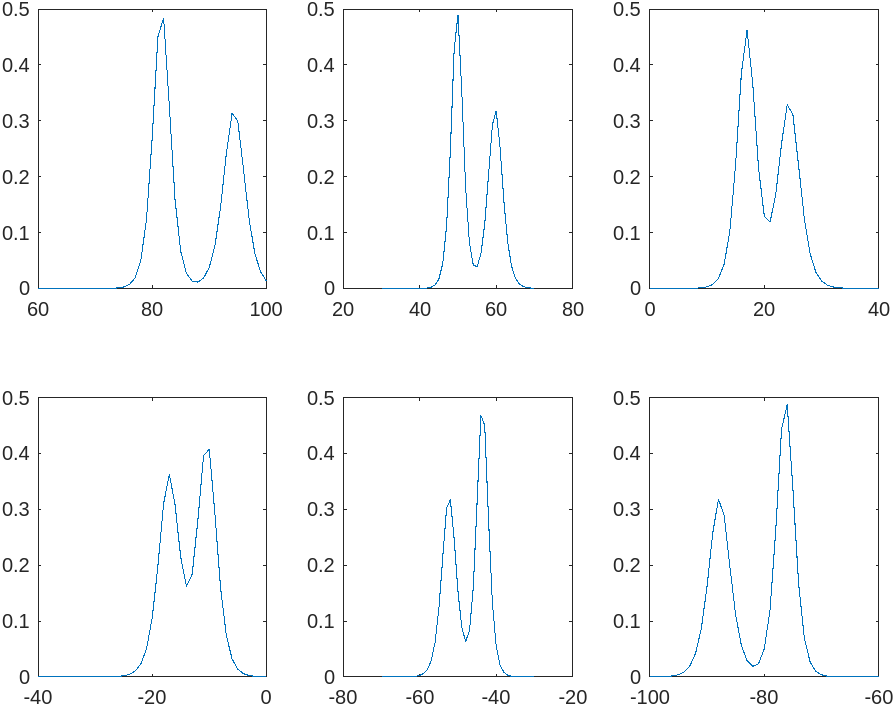}
\caption{\small The process of overtaking collision for two solitary waves at times $t=-100,-60,-20,20,60,100$}
\end{figure}
    
    For $n=3$, the three-soliton solution reads
    \begin{equation}
    \begin{aligned}
u=&2[\log(1+e^{\eta_1}++e^{\eta_2}+e^{\eta_3}+e^{\eta_1+\eta_2+A_{12}}++e^{\eta_1+\eta_3+A_{13}}++e^{\eta_2+\eta_3+A_{23}}\\
    &+e^{\eta_1+\eta_2+\eta_3+A_{12}+A_{13}+A_{23}})]_{xx}
    \end{aligned}
    \end{equation}
    where ${{\eta }_{i}}={{k}_{i}}x+{{w}_{i}}t+\eta _{i}^{\left( 0 \right)}$, $w_{i}^{2}=k_{i}^{2}\left( 1-\frac{1}{3}k_{i}^{2} \right)$, ${{A}_{ij}}=\log \left| -\frac{{{\left( {{w}_{i}}-{{w}_{j}} \right)}^{2}}-{{\left( {{k}_{i}}-{{k}_{j}} \right)}^{2}}+\frac{1}{3}{{\left( {{k}_{i}}-{{k}_{j}} \right)}^{4}}}{{{\left( {{w}_{i}}+{{w}_{j}} \right)}^{2}}-{{\left( {{k}_{i}}+{{k}_{j}} \right)}^{2}}+\frac{1}{3}{{\left( {{k}_{i}}+{{k}_{j}} \right)}^{4}}} \right|$, $i<j$,    $i,j=1,2,3$.

    \subsection{Bilinear B\"acklund Transformation and Lax Pair}
    Next, we search for the bilinear B\"acklund transformation and Lax pair of the good Boussinesq equation(\ref{eq101}). Let $\widetilde{q}=2\log F$, $q=2\log G$ be two different solutions of Eq.(\ref{eq304}), respectively, and  we introduce two new variables,
    \begin{equation}\label{311}
        w=\frac{\widetilde{q}+q}{2}=\log FG, \quad v=\frac{\widetilde{q}-q}{2}=\log \frac{F}{G},
    \end{equation}
    then we associate the two-field condition
        \begin{equation}\label{eq311}
	\begin{split}        
	      E\left( \widetilde{q} \right)-E\left( q \right) &= {{\left( \widetilde{q}-q \right)}_{2t}}-{{\left( \widetilde{q}-q \right)}_{2x}}+{{\left( \widetilde{q}-q \right)}_{2x}}{{\left( \widetilde{q}+q \right)}_{2x}}+\frac{1}{3}{{\left( \widetilde{q}-q \right)}_{4x}} \\
		&= 2{{v}_{2t}}-2{{v}_{2x}}+2{{v}_{2x}}\cdot 2{{w}_{2x}}+\frac{1}{3}\cdot 2{{v}_{4x}} \\
		&= 2{{v}_{2t}}-2{{v}_{2x}}+4{{v}_{2x}}{{w}_{2x}}+\frac{2}{3}\left[ {{\partial }_{x}}{{\mathcal{Y}}_{xxx}}\left( v,w \right)-3{{v}_{2x}}{{w}_{2x}}-3{{v}_{x}}{{w}_{3x}}-3v_{x}^{2}{{v}_{2x}} \right] \\
		&= \frac{2}{3}\left[ {{\partial }_{x}}{{\mathcal{Y}}_{xxx}}\left( v,w \right)+3{{v}_{2x}}{{w}_{2x}}-3{{v}_{x}}{{w}_{3x}}-3v_{x}^{2}{{v}_{2x}}+3{{v}_{2t}}-3{{v}_{2x}} \right] \\
            &= \frac{2}{3}{{\partial }_{x}}{{\mathcal{Y}}_{xxx}}\left( v,w \right)-2{{\partial }_{x}}{{\mathcal{Y}}_{x}}\left( v \right)+R\left( v,w \right)\\
            &= 0,
	\end{split}
    \end{equation}
    where 
    \begin{equation}
        R\left( v,w \right) = 2\operatorname{Wronskian}\left[ {{\mathcal{Y}}_{xx}}\left( v,w \right),{{\mathcal{Y}}_{x}}\left( v \right) \right]+2{{v}_{2t}}.
    \end{equation} 

    In order to decouple the two-field condition (\ref{eq311}) into a pair of constraints, we impose such a constraint which enables us to express $R\left( v,w \right)$ as the $x$-derivative of a combination of $\mathcal{Y}$-polynomials. The simplest possible choice of such a constraint may be 
    \begin{equation}\label{eq313}
        {{\mathcal{Y}}_{2x}}\left( v,w \right)+{{\mathcal{Y}}_{t}}\left( v,w \right)=\lambda,
    \end{equation}
    where $\lambda$ is an arbitrary constant, and later it will be used as the spectral parameter of the Lax pair. In terms of the constraint (\ref{eq313}), $R\left( v,w \right)$ can be expressed as 
    \begin{equation}\label{eq314}
	\begin{split}        
	      R\left( v,w \right) &= 2\left| \begin{matrix}
   \lambda -{{v}_{t}} & -{{v}_{xt}}  \\
   {{v}_{x}} & {{v}_{2x}}  \\
\end{matrix} \right|+2{{v}_{2t}} \\
		&= 2\left( \lambda {{v}_{2x}}-{{v}_{t}}{{v}_{2x}}+{{v}_{x}}{{v}_{xt}} \right)+2\left( -{{w}_{2x,t}}-2{{v}_{x}}{{v}_{xt}} \right) \\
		&= 2\left( \lambda {{v}_{2x}}-{{v}_{t}}{{v}_{2x}}-{{v}_{x}}{{v}_{xt}}-{{w}_{2x,t}} \right)  \\
		&= 2\left( \lambda {{\partial }_{x}}{{\mathcal{Y}}_{x}}\left( v,w \right)-{{\partial }_{x}}{{\mathcal{Y}}_{xt}}\left( v,w \right) \right).
	\end{split}
    \end{equation}
    Then, combining relations (\ref{eq311})-(\ref{eq314}), we deduce a coupled system of $\mathcal{Y}$-polynomials,
    \begin{align}
       & {{\mathcal{Y}}_{2x}}\left( v,w \right)+{{\mathcal{Y}}_{t}}\left( v,w \right) =\lambda, \\
       & \frac{2}{3}{{\mathcal{Y}}_{3x}}\left( v,w \right)-2{{\mathcal{Y}}_{x}}\left( v,w \right)+2\lambda {{\mathcal{Y}}_{x}}\left( v,w \right)-2{{\mathcal{Y}}_{xt}}\left( v,w \right) =0,
    \end{align}
    where the second equation is useful to construct conservation laws later. By application of the proposition \ref{pr204}, the system(3.15)-(3.16) immediately leads to the bilinear B\"acklund transformation
    \begin{align}
       & \left( D_{x}^{2}+{{D}_{t}}-\lambda  \right)F\cdot G=0, \\ 
 & \left( \frac{2}{3}D_{x}^{3}-2{{D}_{x}}+2\lambda {{D}_{x}}-2{{D}_{x}}{{D}_{t}} \right)F\cdot G=0, 
    \end{align}
    where $\lambda$ is a spectral parameter.

    By transformation $v=\ln \psi $, it follows from the formula (\ref{eq210}) and (\ref{eq211}) that 
    \begin{equation*}
    \begin{aligned}
        &{{\mathcal{Y}}_{t}}\left( v,w \right)=\frac{{{\psi }_{t}}}{\psi }, \quad {{\mathcal{Y}}_{x}}\left( v,w \right)=\frac{{{\psi }_{x}}}{\psi }, \quad {{\mathcal{Y}}_{2x}}\left( v,w \right)={{q}_{2x}}+\frac{{{\psi }_{2x}}}{\psi }, \\
        &{{\mathcal{Y}}_{3x}}\left( v,w \right)=3{{q}_{2x}}\frac{{{\psi }_{x}}}{\psi }+\frac{{{\psi }_{3x}}}{\psi }, \quad {{\mathcal{Y}}_{xt}}\left( v,w \right)={{q}_{xt}}+\frac{{{\psi }_{xt}}}{\psi },
    \end{aligned}
    \end{equation*}
    on account of which, the system(3.15)-(3.16) is then linearized into a Lax pair with one parameter $\lambda$
        \begin{align}
  & {{\psi }_{2x}}+{{\psi }_{t}}+\left( u-\lambda  \right)\psi =0, \\ 
 & {{\psi }_{3x}}+\left( 3u+3\lambda -3 \right){{\psi }_{x}}-3\partial _{x}^{-1}{{u}_{t}}\cdot \psi -3{{\psi }_{xt}}=0, 
    \end{align}
    where ${{q}_{2x}}$ has been replaced by $u$. It is easy to check that the integrability condition ${{\psi }_{2x,t}}={{\psi }_{t,2x}}$ is satisfied if $u$ is a solution of the good Boussines equation(\ref{eq101}).
     
    \subsection{Infinite Conservation Laws}
    In this section, we derive the infinitely local conservation laws for good Boussinesq equation(\ref{eq101}) by using binary Bell polynomials. In fact, the conservation laws actually have been hinted in the two-field constraint system(3.15)-(3.16), which can be rewritten in the conserved form
           \begin{align}
  & {{\mathcal{Y}}_{2x}}\left( v,w \right)+{{\mathcal{Y}}_{t}}\left( v,w \right)-\lambda =0, \\ 
 & {{\partial }_{x}}\left[ \frac{1}{3}{{\mathcal{Y}}_{3x}}\left( v,w \right)-{{\mathcal{Y}}_{x}}\left( v,w \right)+\lambda {{\mathcal{Y}}_{x}}\left( v,w \right) \right]-{{\partial }_{t}}{{\mathcal{Y}}_{2x}}\left( v,w \right)=0. 
    \end{align}

    By introducing a new potential function
    \begin{equation}
        \eta =\frac{{{\widetilde{q}}_{x}}-{{q}_{x}}}{2}.
    \end{equation}
    it follows from the relation(\ref{311}) that 
     \begin{equation}\label{eq325}
        {{w}_{x}}={{q}_{x}}+\eta , \quad {{v}_{x}}=\eta , \quad {{v}_{t}}=\partial _{x}^{-1}{{\eta }_{t}}. 
    \end{equation}
    Substituting (\ref{eq325}) into (3.22) and (3.23), we get a Riccati-type equation
     \begin{equation}\label{eq326}
        {{\eta }_{x}}+{{\eta }^{2}}+{{q}_{2x}}+\partial _{x}^{-1}{{\eta }_{t}}=\lambda, 
    \end{equation}
    and a divergence-type equation
    \begin{equation}\label{eq327}
        \frac{1}{3}{{\eta }_{3x}}+\left( 2\lambda -1 \right){{\eta }_{x}}-2{{\eta }^{2}}{{\eta }_{x}}-{{\eta }_{x}}\cdot \partial _{x}^{-1}{{\eta }_{t}}-\eta {{\eta }_{t}}+\partial _{x}^{-1}{{\eta }_{tt}}=0, 
    \end{equation}
    where we have used Eq.(\ref{eq326}) to get Eq.(\ref{eq327}).
    Suppose that $\lambda ={{\varepsilon }^{2}}$, to proceed, inserting the expansion 
    \begin{equation}\label{eq328}
        \eta =\varepsilon +\sum\limits_{n=1}^{\infty }{{{I}_{n}}\left( q,{{q}_{x}},\cdots  \right)}{{\varepsilon }^{-n}}
    \end{equation}
    into Eq.(\ref{eq326}) and equating the coefficients for power of $\varepsilon $, we then obtain the recursion relations for the conserved densities ${{I}_{n}}$
    \begin{align}
    & {{I}_{1}}=-\frac{1}{2}{{q}_{2x}}=-\frac{1}{2}u, \\ 
  & {{I}_{2}}=-\frac{1}{2}{{I}_{1,x}}-\frac{1}{2}\partial _{x}^{-1}{{I}_{1}}_{,t}=\frac{1}{4}{{u}_{x}}+\frac{1}{4}\partial _{x}^{-1}{{u}_{t}}, \\ 
  & \cdots, \nonumber \\
  & {{I}_{n+1}}=-\frac{1}{2}\left( {{I}_{n,x}}+\sum\limits_{k=1}^{n}{{{I}_{k}}\cdot {{I}_{n-k}}} \right)-\frac{1}{2}\partial _{x}^{-1}{{I}_{n,t}}.  
   \end{align}
   Substituting Eq.(\ref{eq328}) into Eq.(\ref{eq327}) provides us infinite consequence of conservation laws
   \begin{equation}
       {{I}_{n,t}}+{{F}_{n,x}}=0, \quad n=1,2,\cdots
   \end{equation}
   where the conserved densities ${{I}_{n}}\left( n=1,2,\cdots  \right)$ are given by formula (3.31) and the fluxes ${{F}_{n}}\left( n=1,2,\cdots  \right)$ are given by recursion formulas explicitly
      \begin{align}
   &{{F}_{1}}=2{{I}_{2}}=\frac{1}{2}{{u}_{x}}+\frac{1}{2}\partial _{x}^{-1}{{u}_{t}}, \\ 
  &{{F}_{2}}=-\frac{1}{3}{{I}_{1,xx}}-\left( 2\lambda -1 \right){{I}_{1}}+2\left( I_{1}^{2}+{{I}_{3}} \right)-\partial _{x}^{2}{{I}_{1,tt}} \nonumber  \\ 
  & \quad=-\frac{1}{12}{{u}_{xx}}+\left( \lambda -\frac{1}{2} \right)u+\frac{1}{4}{{u}^{2}}-\frac{1}{2}{{u}_{t}}+\frac{1}{4}\partial _{x}^{2}{{u}_{tt}}, \\
  & \cdots, \nonumber \\
  &{{F}_{n+1}}=-\frac{1}{3}{{I}_{n,xx}}-\left( 2\lambda -1 \right){{I}_{n}}+2\left[ {{I}_{n+2}}+\sum\limits_{k=1}^{n}{{{I}_{k}}\cdot {{I}_{n+1-k}}+\frac{1}{3}\sum\limits_{i+j+k=n}{{{I}_{i}}{{I}_{j}}{{I}_{k}}}} \right] \nonumber \\
 & \quad\quad+\sum\limits_{l=1}^{n-1}{{{I}_{l}}}\partial _{x}^{-1}{{I}_{n-l,t}}-\partial _{x}^{-2}{{I}_{n,tt}}.
\end{align}
We present recursion formulas for generating an infinite sequence of coservation laws; The first few conserved densities and associated fluxes are explicitly given. The first equation of conservation law is exactly the good Boussinesq equation(\ref{eq101}).

    \section{$n$ order $\operatorname{Wronskian}$ determinant solution}
    In this section, a $n$ order $\operatorname{Wronskian}$ determinant solution will be established to Eq.(\ref{eq101}) via the $\operatorname{Wronskian}$ technique. To use this technique, we adopt the following helpful notation.
    \begin{equation}\label{eq401}
        \begin{split}    
          W\left( {{\varphi }_{1}},\cdots ,{{\varphi }_{n}} \right) & =\left| \begin{matrix}
   {{\varphi }_{1}} & \varphi _{1}^{\left( 1 \right)} & \cdots  & \varphi _{1}^{\left( n-1 \right)}  \\
   {{\varphi }_{1}} & \varphi _{2}^{\left( 1 \right)} & \cdots  & \varphi _{2}^{\left( n-1 \right)}  \\
   \cdots  & \cdots  & \cdots  & \cdots   \\
   {{\varphi }_{n}} & \varphi _{n}^{\left( 1 \right)} & \cdots  & \varphi _{n}^{\left( n-1 \right)}  \\ 
\end{matrix} \right|   \\
        & = \left| \varphi ,{{\varphi }^{\left( 1 \right)}},\cdots ,{{\varphi }^{\left( n-1 \right)}} \right|  \\
        & = \left| 0,1,\cdots ,n-1 \right| \\
        & =\left| \widehat{n-1} \right|,
        \end{split}
      \end{equation}
      where $\varphi _{j}^{(k)}=\frac{{{\partial }^{k}}{{\varphi }_{j}}}{\partial {{x}^{k}}}$, $\varphi ={{\left( {{\varphi }_{1}},\cdots ,{{\varphi }_{n}} \right)}^{T}}$, ${{\varphi }^{\left( k \right)}}={{\left( \varphi _{1}^{\left( k \right)},\cdots ,\varphi _{n}^{\left( k \right)} \right)}^{T}}$, $k=1,2,\cdots ,n-1$.
      
      In order to obtain the $\operatorname{Wronskian}$ solution, the following lemmas are needed.
      
      \begin{lemma}
          \begin{equation}
        \left| M,a,b \right|\cdot \left| M,c,d \right|-\left| M,a,c \right|\cdot \left| M,b,d \right|+\left| M,a,d \right|\cdot \left| M,b,c \right|=0,
    \end{equation}
     where M is an $n\times \left( n-2 \right)$ matrix, $a, b, c$ and $d$ represent $n$ column vectors.
      \end{lemma}

      \begin{lemma}
          \begin{equation}\label{eq403}
            \sum\limits_{k=1}^{n}{{{\left| A \right|}_{kl}}}=\sum\limits_{k=1}^{n}{{{\left| A \right|}^{kl}}=}\sum\limits_{i,j=1}^{n}{{{A}_{ij}}\frac{{{\partial }^{l}}{{a}_{ij}}}{\partial {{x}^{l}}}},
        \end{equation}
        where $A={{\left( {{a}_{ij}} \right)}_{n\times n}}$, and ${{\left| A \right|}_{kl}}$, ${{\left| A \right|}^{kl}}$ and ${{A}_{ij}}$ denote the determinant resulting from $\left| A \right|$ with its $k$th row differentiated l times with respect to $x$,  the determinant resulting from $\left| A \right|$ with its $k$th column differentiated l times with respect to $x$, and the co-factor of ${{a}_{ij}}$, respectively.
      \end{lemma}

      In particular, choose
      \begin{equation*}
        \left| A \right|=\left| \widehat{n-1} \right|=\left| \varphi ,{{\varphi }^{\left( 1 \right)}},\cdots ,{{\varphi }^{\left( n-1 \right)}} \right|=\left| 0,1,\cdots ,n-1 \right|=\left| \begin{matrix}
   \varphi _{1}^{\left( 0 \right)} & \varphi _{1}^{\left( 1 \right)} & \cdots  & \varphi _{1}^{\left( n-1 \right)}  \\
   \vdots  & \vdots  & \vdots  & \vdots   \\
   \varphi _{i}^{\left( 0 \right)} & \varphi _{i}^{\left( 1 \right)} & \cdots  & \varphi _{i}^{\left( n-1 \right)}  \\
   \vdots  & \vdots  & \vdots  & \vdots   \\
   \varphi _{n}^{\left( 0 \right)} & \varphi _{n}^{\left( 1 \right)} & \cdots  & \varphi _{n}^{\left( n-1 \right)}  \\
\end{matrix} \right|,
    \end{equation*}
    and use the equality(\ref{eq403}) with $l=3$, then we obtain the equality as follows.
      \begin{equation}\label{eq404}
       \left| \widehat{n-4},n-2,n-1,n \right|-\left| \widehat{n-3},n-1,n+1 \right|+\left| \widehat{n-2},n+2 \right|-\frac{3}{4}\left| \widehat{n-2},n \right|=0.
   \end{equation}
   Differentiating with respect to x yields 
   \begin{equation}\label{eq405}
   \begin{aligned}
       \left| \widehat{n-5},n-3,n-2,n-1,n \right|-\left| \widehat{n-3},n,n+1 \right|  +\left| \widehat{n-2},n+3 \right|
      \\  -\frac{3}{4}\left| \widehat{n-3},n-1,n \right|
       -\frac{3}{4}\left| \widehat{n-2},n+1 \right|=0.
       \end{aligned}
   \end{equation}
   The identities (\ref{eq404}) and (\ref{eq405}) are very useful in constructing the $\operatorname{Wronskian}$solution of the equation (\ref{eq101}).

   Next, we construct $\operatorname{Wronskian}$ condition of good Boussinesq equation (\ref{eq101}) by virtue of the Lax pairs (3.20) and (3.21) obtained in section 3.3. We choose $u=0$, $\lambda=0$, then the Lax pair (3.20) and (3.21) reduces to
   \begin{equation}\label{eq406}
        {{\psi }_{t}}=-{{\psi }_{2x}}, \quad {{\psi }_{3x}}=\frac{3}{4}{{\psi }_{x}}.
    \end{equation}
    Therefore, choose ${{\psi }_{j}}\in {{C}^{\infty }}\left( \Omega  \right)$ satisfying 
    \begin{equation}\label{eq407}
        {{\psi }_{j,t}}=-{{\psi }_{j,xx}}, \quad{{\psi }_{j,xxx}}=\frac{3}{4}{{\psi }_{j,x}}, \quad j=1,2,\cdots ,n.
    \end{equation}

    \begin{theorem}
      Assume that a group of function ${{\psi }_{j}}\left( x,t \right)$, $j=1,2,\cdots ,n$, satisfies 
      \begin{equation}
        {{\psi }_{j,t}}=-{{\psi }_{j,xx}}, \quad{{\psi }_{j,xxx}}=\frac{3}{4}{{\psi }_{j,x}}, \quad j=1,2,\cdots ,n,
    \end{equation}
    simultaneously. Then $F=W\left( {{\psi }_{1}},{{\psi }_{2}},\cdots ,{{\psi }_{n}} \right)$ defined by (\ref{eq401}) solves the bilinear Boussinesq equation(\ref{eq306}).
    \end{theorem}
    \begin{proof}
        It only needs to prove that $F=W\left( {{\psi }_{1}},{{\psi }_{2}},\cdots ,{{\psi }_{n}} \right)$ solves bilinear equation(\ref{eq306}). To complete the proof, the bilinear equation(\ref{eq306}) can be reduced to determinant identity. Therefore, the bilinear equation(\ref{eq306}) is rewritten as
         \begin{equation}\label{eq409}
        \Delta \equiv 2{{F}_{2t}}F-2F_{t}^{2}-2F{{F}_{2x}}+2F_{x}^{2}+\frac{2}{3}F{{F}_{4x}}-\frac{8}{3}{{F}_{x}}{{F}_{3x}}+2F_{2x}^{2}=0.
    \end{equation}
    Using the property of $\operatorname{Wronskian}$ determinant, calculate the derivatives required by the equation(\ref{eq409})
                \begin{align*}
  & {{F}_{x}}=\left| \widehat{n-2},n \right|, \\ 
 & {{F}_{2x}}=\left| \widehat{n-3},n-1,n \right|+\left| \widehat{n-2},n+1 \right|, \\ 
 & {{F}_{3x}}=\left| \widehat{n-4},n-2,n-1,n \right|+2\left| \widehat{n-3},n-1,n+1 \right|+\left| \widehat{n-2},n+2 \right|, \\ 
 & {{F}_{4x}}=\left| \widehat{n-5},n-3,n-2,n-1,n \right|+3\left| \widehat{n-4},n-2,n-1,n+1 \right|+2\left| \widehat{n-3},n,n+1 \right| \\
 &\quad\quad\quad +3\left| \widehat{n-3},n-1,n+2 \right|+\left| \widehat{n-2},n+3 \right|.
              \end{align*}
        Using the formula (\ref{eq406}), we can get
        \begin{equation*}
            {{F}_{t}}= \left| n-3,n-1,n \right|-\left| n-2n+1 \right|,
        \end{equation*}      
    \begin{equation*}
    \begin{aligned}
        {{F}_{tt}}=&\left| n-5,n-3,n-2,n-1,n \right|-\left| n-4,n-2,n-1,n+1 \right|+2\left| n-3,n,n+1 \right| \\
  &-\left| n-3,n-1,n+2 \right|+\left| n-2,n+3 \right|. 
  \end{aligned}
    \end{equation*}
    Substitute the above equations into the left hand of the formula (\ref{eq409}), we can get
            \begin{equation}
     \begin{split}        
	      \Delta  =&2F\left( {{F}_{2t}}-{{F}_{2x}}+\frac{1}{3}{{F}_{4x}} \right)-2F_{t}^{2}+2{{F}_{x}}\left( {{F}_{x}}-\frac{4}{3}{{F}_{3x}} \right)+2F_{2x}^{2} \\ 
   =&2\left| n-1 \right|\left( \left| n-5,n-3,n-2,n-1,n \right|-\left| n-4,n-2,n-1,n+1 \right|+2\left| n-3,n,n+1 \right| \right. \\ 
 & -\left| n-3,n-1,n+2 \right|+\left| n-2,n+3 \right|-\left| n-3,n-1,n \right|-\left| n-2,n+1 \right| \\
 &+\frac{1}{3}\left| n-5,n-3,n-2,n-1,n \right|  
  +\left| n-4,n-2,n-1,n+1 \right|+\frac{2}{3}\left| n-3,n,n+1 \right|\\
  &+\left| n-3,n-1,n+2 \right|+ \frac{1}{3}\left| n-2,n+3 \right| ) -2{{\left( \left| n-3,n-1,n \right|-\left| n-2,n+1 \right| \right)}^{2}} \\
  &+2\left| n-2,n \right|( \left| n-2,n \right|
  -\frac{4}{3}\left| n-4,n-2,n-1,n \right| )  
  -\frac{8}{3}\left| n-3,n-1,n+1 \right|\\
  &- \frac{4}{3}\left| n-2,n+2 \right| ) +2{{\left( \left| n-3,n-1,n \right|+\left| n-2,n+1 \right| \right)}^{2}}  \\
  =&2\left| n-1 \right|( \frac{4}{3} \left| n-5,n-3,n-2,n-1,n \right|+\frac{8}{3}\left| n-3,n,n+1 \right|+\frac{4}{3}\left| n-2,n+3 \right| \\
  &-\left| n-3,n-1,n \right|  
  \left. -\left| n-2,n+1 \right| \right)+8\left| n-3,n-1,n \right|\cdot \left| n-2,n+1 \right|+2\left| n-2,n \right|\\
  &\left( \left| n-2,n \right| \right.
  -\frac{4}{3}\left| n-4,n-2,n-1,n \right|  
  -\frac{8}{3}\left| n-3,n-1,n+1 \right|-\frac{4}{3}\left| n-2,n+2 \right| ).  
    \end{split}
    \end{equation}
    Using the identities (\ref{eq404}) and (\ref{eq405}), the following Pl\"ucker relation can be obtained
         \begin{equation}
	\begin{split}        
	      \Delta =& 8\left| n-1 \right|\cdot \left| n-3,n,n+1 \right|+8\left| n-3,n-1,n \right|\cdot \left| n-2,n+1 \right|\\
       &-8\left| n-2,n \right|\cdot \left| n-3,n-1,n+1 \right| \\ 
  =&4\left| \begin{matrix}
   n-3 & 0 & n-2 & n-1 & n & n+1  \\
   0 & n-3 & n-2 & n-1 & n & n+1  \\
\end{matrix} \right| \\ 
  \equiv& 0.
		\end{split}
    \end{equation}
    \end{proof}
    Theorem 4.3 tells us that if a group of functions ${{\psi }_{j}}\left( x,t \right)$ ($j=1,2,\cdots ,n$) satisfy conditions in (\ref{eq407}), then we can get a solution $F=W\left( {{\psi }_{1}},{{\psi }_{2}},\cdots ,{{\psi }_{n}} \right)$ to the bilinear Boussinesq equation(\ref{eq306}). Thus, the $\operatorname{Wronskian}$ determinant solution of the original equation (\ref{eq101}) is 
    \begin{equation}
        u=2{{\left( \log F \right)}_{xx}}=2{{\left( \log W\left( {{\psi }_{1}},{{\psi }_{2}},\cdots ,{{\psi }_{n}} \right) \right)}_{xx}}.
    \end{equation}
      
\section{Concluding remarks}
    The present investigation has been carried out on the good Boussinesq equation (\ref{eq101}). With binary Bell polynomials, the corresponding bilinear representation, bilinear B\"acklund transformation, Lax pair and infinite conservation laws are obtained directly and naturally. Combined with the perturbation expansion method, the $n$-soliton solution has also been derived.

    Finally, inspired by the construction of $\operatorname{Wronskian}$ condition for the KdV equation, this paper tries to construct the $\operatorname{Wronskian}$ condition for the good Boussinesq equation from the Bell polynomials. We believe that there are still many deep relations between Bell polynomials and $\operatorname{Wronskian}$ technique, which still remain open and become an issue for future research to explore.

   \section*{Acknowledgments}
    We express our sincere thanks to Prof.E.G. Fan and Prof.Y.Chen for his valuable guidence and advice.
   {\linespread{1.0} \selectfont

}


\addcontentsline{toc}{section}{References}
\begin{thebibliography}{9}
\bibitem{Drazin} Drazin P G, Johnson R S. Solitons: an introduction [M]. Cambridge university press, 1989.

\bibitem{AAC} Ablowitz M J, Ablowitz M A, Clarkson P A. Solitons, nonlinear evolution equations and inverse scattering [M]. Cambridge university press, 1991.

\bibitem{Hirota} Hirota R. The direct method in soliton theory [M]. Cambridge University Press, 2004.

\bibitem{GLN} Gilson C, Lambert F, Nimmo J, et al. On the combinatorics of the Hirota D-operators [J]. Proceedings of the Royal Society of London. Series A: Mathematical, Physical and Engineering Sciences, 1996, 452 (1945): 223-234.

\bibitem{LLS} Lambert F, Loris I, Springael J. Classical Darboux transformations and the KP hierarchy [J]. Inverse Problems, 2001, 17 (4): 1067-1074.

\bibitem{LS} Lambert F, Springael J. Soliton equations and simple combinatorics [J]. Acta Applicandae Mathematicae, 2008, 102: 147-178.

\bibitem{egfan} Fan E G. The integrability of nonisospectral and variable-coefficient KdV equation with binary Bell polynomials [J]. Physics Letters A, 2011, 375 (3): 493-497.

\bibitem{FH} Fan E G, Hon Y C. Super extension of Bell polynomials with applications to supersymmetric equations [J]. Journal of Mathematical Physics, 2012, 53 (1): 013503.

\bibitem{MaWX} Ma W X. Bilinear equations, Bell polynomials and linear superposition principle [C]. Journal of Physics: Conference Series. IOP Publishing, 2013, 411 (1): 012021.

\bibitem{MWCY} Miao Q, Wang Y, Chen Y, et al. PDEBellII: A Maple package for finding bilinear forms, bilinear Bäcklund transformations, Lax pairs and conservation laws of the KdV-type equations [J]. Computer Physics Communications, 2014, 185 (1): 357-367.

\bibitem{Hirota1971} Hirota R. Exact solution of the Korteweg-de Vries equation for multiple collisions of solitons [J]. Physical Review Letters, 1971, 27 (18): 1192.

\bibitem{FN} Freeman N C, Nimmo J J C. Soliton solutions of the Korteweg-de Vries and the Kadomtsev-Petviashvili equations: the Wronskian technique [J]. Proceedings of the Royal Society of London. A. Mathematical and Physical Sciences, 1983, 389 (1797): 319-329.

\bibitem{Matveev} Matveev V B. Generalized Wronskian formula for solutions of the KdV equations: first applications [J]. Physics Letters A, 1992, 166 (3-4): 205-208.

\bibitem{MY} Ma W X, You Y. Solving the Korteweg-de Vries equation by its bilinear form: Wronskian solutions [J]. Transactions of the American mathematical society, 2005, 357 (5): 1753-1778.

\bibitem{Whitham} Whitham G B. Linear and nonlinear waves [M]. John Wiley Sons, 2011.

\bibitem{Xu} Xu L, Auston D H, Hasegawa A. Physical Review A, 1992, 45: 3184-93

\bibitem{Karpman} Karpman V I. Nonlinear waves in dispersive media: International series of monographs in natural philosophy [M]. Elsevier, 2016.

\bibitem{Turitsyn} Turitsyn S K, Falkovich G E. Stability of magnetoelastic solitons and self-focusing of sound in antiferromagnetics [J]. Zhurnal Eksperimentalnoi I Teoreticheskoi Fiziki, 1985, 89 (1): 258-270.

\bibitem{Matveev} Matveev V B, Salle M A. Darboux transformations and solitons [M]. Berlin: Springer, 1991.

\bibitem{HirotaR} Hirota R. Exact N‐soliton solutions of the wave equation of long waves in shallow‐water and in nonlinear lattices [J]. Journal of Mathematical Physics, 1973, 14 (7): 810-814.

\bibitem{Chunxia} Li C X, Ma W X, Liu X J, et al. Wronskian solutions of the Boussinesq equation-solitons, negatons, positons and complexitons [J]. Inverse Problems, 2007, 23 (1): 279-296.

\bibitem{Nguyen} Nguyen L T K. Soliton solution of good Boussinesq equation [J]. Vietnam Journal of Mathematics, 2016, 44: 375-385.



\end{thebibliography}
\end{document}